\PassOptionsToPackage{hyphens}{url}
\documentclass[a4paper,UKenglish,cleveref,autoref,]{lipics-v2021}
\nolinenumbers
\input{macros}

\title{Pod-Deployability in Kubernetes with Inter-Pod Affinity Constraints is PSPACE-Complete}
\titlerunning{Kubernetes Pod-Deployability is PSPACE-Complete}
\author{Saverio Giallorenzo}
  {Universit\`a di Bologna, Italy \and INRIA, Sophia-Antipolis, France}
  {}{}{}
\author{Jacopo Mauro}
  {University of Southern Denmark, Denmark}
  {}{}{}
\author{Gianluigi Zavattaro}
  {Universit\`a di Bologna, Italy \and INRIA, Sophia-Antipolis, France}
  {}{}{}
\authorrunning{S. Giallorenzo, J. Mauro, and G. Zavattaro}
\Copyright{Saverio Giallorenzo, Jacopo Mauro, and Gianluigi Zavattaro}
\EventEditors{John Q. Open and Joan R. Access}
\EventNoEds{2}
\EventLongTitle{XXXX}
\EventShortTitle{XXXX 2027}
\EventAcronym{XXXX}
\EventYear{2027}
\EventDate{July 2027}
\EventLocation{TBD}
\EventLogo{}
\SeriesVolume{XXX}
\ArticleNo{XX}

\begin{document}
\maketitle
\begin{abstract}
Kubernetes is the de-facto platform for container orchestration. Its scheduler
combines resource capacities with label-based affinity and anti-affinity rules,
and the interaction of these features can make the eventual placement of a pod.

In this paper, we study the \emph{pod-deployability} problem: given an initial cluster, a pod type,
and a designated node, does some legal sequence of pod deployments and
deletions cover the target pair?
We give three complexity results. First, when dynamic constraints
contain no affinity (anti-affinity is allowed),
pod-deployability is decidable in polynomial time. Second, required affinity
together with required anti-affinity makes the problem \(\PSPACE\)-complete. Third, required affinity alone is already
enough for \(\PSPACE\)-completeness on a single node with one scalar capacity.
The lower bounds encode, respectively,
1-safe Petri-net coverability and bounded black pebbling. These results isolate
two independent sources of state-space complexity in Kubernetes scheduling:
logical exclusion and resource-bounded prerequisite management.
\end{abstract}
\newpage
\section{Introduction}\label{sec:introduction}

Kubernetes \cite{kubernetes_documentation} is the de-facto standard platform for
orchestrating containerized applications within a cluster of machines
\cite{carrion2023kubernetes}---around 82\% of container-using organizations
worldwide adopt it and major cloud providers (like AWS, Google, and Azure)
offer it as a managed service \cite{cncf2026survey, mordor2026market}.

Crucial terms in Kubernetes' lingo are \emph{nodes} and \emph{pods}. The former
are the machines participating in a cluster. Each node contributes computing
resources, such as CPU and memory, and may carry labels describing properties
such as its location, hardware, or security domain. Pods are the smallest unit
managed by Kubernetes and a pod groups one or more containers placed on the same
node and share runtime resources.

At the heart of Kubernetes lies its scheduler. When the platform receives the
command to deploy a given workload, it creates one or more pods with no assigned
node, which triggers the scheduler to evaluate the cluster and select a suitable
node for each one to run on. A node is suitable only if it satisfies all the
pod's hard scheduling constraints. These constraints include, among others,
sufficient resource capacity, requirements on node labels, and relationships
with pods that are already running. In particular, \emph{inter-pod affinity} can
require a pod to be placed close to other pods, while \emph{inter-pod
anti-affinity} can require separation from them---where ``closeness'' is
determined by a topology domain, e.g., ranging from single nodes to whole
availability zones. Although Kuberbetes evaluates each scheduling constraint
locally at each pod placement, their combined effect over a sequence of pod
deployments and deletions is not local.

In this paper, we focus on the decidability and complexity of a basic question:
\emph{can a pod of type \(t\) ever be deployed on node \(n\)}? Answering this
question is practically important, for instance, a positive answer may expose an
unintended crossing of a security boundary---e.g., a pod processing private data
might be deployable on a node belonging to a public or otherwise untrusted
worker pool. We call this decision problem \emph{pod-deployability}. Formally,
we frame the problem as the coverability of a type-node coordinate and ask
whether some reachable cluster configuration contains at least one instance of
the target type-node.%

\noindent
\textbf{Contributions.}
We formalise the Kubernetes scheduling semantics (\Cref{sec:model}) and use it
to give three main results that expose a sharp complexity boundary.
\begin{enumerate}
\item With anti-affinity and no capacities nor affinity, pod-deployability
reduces to one feasibility check in the empty configuration and is decidable in
time polynomial in the number of nodes and the total encoding length of the
scheduling constraints (\Cref{sec:polynomial}).
\item With affinity and anti-affinity and no capacities, pod-deployability is
\(\PSPACE\)-complete. The proof uses a four-phase protocol to simulate 1-safe
Petri-net transitions (\cref{sec:logical-pspace}).
\item With affinity alone, pod-deployability is \(\PSPACE\)-complete already for
one node and one scalar capacity---reduction from black-pebbling under a
matching space bound.
(\cref{sec:capacity-pspace}).
\item 
We mechanized all reductions and formally verified them in Lean, establishing a
machine-checked correspondence between pod-deployability and the corresponding
hardness sources.
\end{enumerate}
\noindent
\textbf{Structure of the paper.}
In \Cref{sec:model}, we formalize Kubernetes' scheduler semantics. Then, we
formulate pod-deployability as a coverability problem in \cref{sec:problem},
proving that it is in \(\PSPACE\). We study the polynomial fragment without
affinity in \cref{sec:polynomial} and present the two \(\PSPACE\)-completeness
results in \cref{sec:logical-pspace,sec:capacity-pspace}. We discuss related
work and conclude in \cref{sec:related}. The appendix includes detailed proofs
of all results and the encodings of the hardness proofs, which we formalised in
Lean and make available at \url{https://doi.org/10.5281/zenodo.22025656}.

\blfootnote{
\textbf{Disclosure of AI Usage.}
We used AI-based tools to assist in preparing this paper. Specifically, such tools helped refine the clarity and readability of the prose, and assisted in proving theorems and refactoring code within the Lean formalization. We reviewed and independently verified all AI-assisted output. We take full responsibility for the correctness of the results of the paper.}
\section{Kubernetes Scheduling Made Formal}
\label{sec:model}

We present Kubernetes and formalize its standard pod scheduling policy.
Kubernetes is a container orchestration platform that manages applications
deployed as collections of \emph{pods}. A Kubernetes cluster consists of a set
of worker nodes together with a control plane responsible for maintaining the
desired cluster state. Among its components, the Kubernetes scheduler determines
where newly created pods should run.

When a pod is created without a node assignment, it remains pending until the
scheduler selects a suitable node and records a binding
\cite{k8s-scheduler,k8s-framework}. Scheduling proceeds conceptually in two
main stages. During \emph{filtering}, nodes that violate mandatory scheduling
constraints are discarded. During \emph{scoring}, the remaining nodes are
ranked according to preferences such as resource usage or preferred
affinity. After a binding is committed, the kubelet on the selected node
starts and supervises the pod's containers.

Since this work is concerned with the feasibility of placements---rather than
with the scheduler's choice among feasible nodes---we model only \emph{hard}
scheduling constraints. Indeed, soft constraints affect which feasible node is
preferred but cannot make an otherwise infeasible placement legal. We also
abstract from implementation details such as scheduler concurrency, internal
caches, transient API failures, preemption, and binding plugins. The resulting
model retains the cluster information relevant to hard placement feasibility and
represents successful scheduling nondeterministically: any node satisfying all
hard constraints may be selected.

\subsection{Pods, Nodes, and Cluster Configurations}

We introduce the formal model alongside the Kubernetes objects that it
represents. Let \(\PodTypes\), \(\Nodes\), and \(\Resources\) be finite sets of
pod types, nodes, and resource dimensions, respectively. Let also
$\mathsf{LabelKey}$ and $\mathsf{LabelValue}$ be a set of labels and
corresponding values, respectively.

A \emph{pod type} represents the scheduling-relevant information shared by
instances created from the same pod template. In particular, each
\(t\in\PodTypes\) has a label map
\(
\lambda_T(t):
\mathsf{LabelKey}\rightharpoonup\mathsf{LabelValue}
\)
and a resource-request vector
\(
\texttt{res}(t)\in\mathbb N^{\Resources}.
\)
Similarly, each node \(n\in\Nodes\) has a label map
\(
\lambda_N(n):
\mathsf{LabelKey}\rightharpoonup\mathsf{LabelValue}
\)
and an allocatable capacity vector
\(
C(n)\in\mathbb N^{\Resources}.
\)

The scheduler reasons about pod resource \emph{requests},
and resources include dimensions like CPU and memory. We denote which pods instances run on which cluster nodes as
\[
\sigma:\PodTypes\times\Nodes\longrightarrow\mathbb N,
\]
where \(\sigma(t,n)\) is the number of running instances of type \(t\) on node
\(n\).

Therefore, the total resources requested on a node \(n\) in configuration \(\sigma\)
are 
\[
\used(\sigma,n)
=
\sum_{t\in\PodTypes}
\sigma(t,n) \; \texttt{res}(t).
\]
Consequently, placing an additional instance of type \(t\) on \(n\) respects
the node's capacity 
when
\[
\used(\sigma,n)+\texttt{res}(t)\leq C(n),
\]
with component-wise vector addition and comparison.

We write \(\zero\) for the empty configuration, defined by \(\zero(t,n)=0\) for
every \(t\in\PodTypes,\ n\in\Nodes\). We further write \(\sigma+(t,n)\) for the
configuration obtained by incrementing \(\sigma(t,n)\) by one, and
\(\sigma-(t,n)\) for the corresponding decrement.

Kubernetes uses labels and selectors to express relationships between pods.
A selector \(s\) is a predicate over pod labels. We write
\[
\sem{s}
=
\{t\in\PodTypes
  \mid \lambda_T(t)\text{ satisfies }s\}
\]
for the set of pod types matched by \(s\).

Scheduling constraints can also refer to the topology of the cluster. Kubernetes
represents topology through node labels. A topology key \(\tau\), e.g., the
hostname or availability zone of a node, partitions the nodes according to the
corresponding label value.
For a node \(n\), we define its topology domain with respect to \(\tau\) by
\[
D_\tau(n)
=
\{\ m\in\Nodes
  \ \mid\
  \lambda_N(m)(\tau)=\lambda_N(n)(\tau)\ \}.
\]
Informally, we can view \(D_\tau(n)\) as the \(n\)-induced domain under topology
key \(\tau\). Thus, if \(\tau\) denotes the hostname label, \(D_\tau(n)\)
typically contains only \(n\), while if \(\tau\) denotes a zone, \(D_\tau(n)\)
contains all the nodes in \(n\)'s zone. For well-formedness, we assume that all
nodes have a value for every topology key, so that \(\lambda_N(m)(\tau)\) is
always defined.

\subsection{Static Constraints}

Kubernetes enforces several hard constraints that determine whether a node is a
valid candidate for a pod independently of the current cluster configuration: i)
node affinity constraints use labels to restrict a pod to nodes satisfying
specified label requirements (e.g., a pod may require a node labelled as
belonging to a particular zone); ii) taints and tolerations allow nodes to repel
pods that do not explicitly tolerate particular node conditions or
administrative restrictions (e.g., a node reserved for a particular workload may
be tainted so that only pods explicitly configured to tolerate that taint can be
placed on it).

We collect these configuration-independent checks into the predicate
$\StaticOK(t,n)$. For a pod type $t$ and a candidate node $n$, evaluating
$\StaticOK(t,n)$ amounts to matching the label requirements and tolerations of $t$
against the labels and taints of $n$, and takes time linear in the total number
of these label requirements, tolerations, labels, and taints.
\medskip

\noindent
\textbf{Other Scheduling Constraints.}
While minimal, the combination of resource and pod-node constraints
can capture additional ad-hoc constraints
enforced by Kubernetes, for instance, the availability of
persistent volumes and host network ports~\cite{k8s-config}.

Regarding volumes, we can represent volume-related placement constraints by
combining $\StaticOK(t,n)$ with resource capacities, which captures both
compatibility requirements between a pod and the volumes available on a
candidate node and quantitative restrictions, like the available storage or
number of volume slots. Regarding network ports, we can model port-related
placement constraints using resource capacities. In particular, we can capture a
port as a distinct resource with unit capacity on each node, while a pod
requesting that port consumes one unit of the corresponding resource. The
capacity constraint then ensures that no two pods scheduled on the same node can
claim the same port.

\subsection{Inter-Pod Affinity and Anti-Affinity Constraints}

Inter-pod affinity and anti-affinity constrain a pod according to the location
of other pods. Required affinity permits a placement only when a pod matching a
given selector (\(s\)) is already present in a specified topology domain
(determined by \(\tau\)). Required anti-affinity has the opposite effect: it
prevents the incoming pod from being placed in a domain containing a matching
pod.

We represent an affinity or anti-affinity term by a pair \((s,\tau)\), where
\(s\) is a pod selector and \(\tau\) a topology key. For each pod type \(t\),
let \(A^+(t)\) and \(A^-(t)\) denote its sets of required affinity and required
anti-affinity terms, respectively.

We define predicate \(\mathsf{present}\) to determine whether a matching pod is present in a given \(n\)-induced, \(\tau\)-determined domain
\[
\mathsf{present}_\sigma(s,\tau,n)
\iff
\exists u\in\sem{s}\;
\exists m\in D_\tau(n):
\sigma(u,m)>0.
\]

The required affinity and anti-affinity constraints of an incoming pod of type
\(t\) on node \(n\) in configuration \(\sigma\) are satisfied when predicate
\(\mathsf{AffinityOK}_\sigma(t,n)\) holds.

\[
\mathsf{AffinityOK}_\sigma(t,n)\iff\;
  \overbrace{
  \bigwedge_{(s,\tau)\in A^-(t)}
  \hspace{-.5em}
  \neg\mathsf{present}_\sigma(s,\tau,n)}^{\text{required anti-affinity}} 
  \quad\land\
  \overbrace{
  \hspace{-2em}\bigwedge_{\substack{(s,\tau)\in A^+(t)\\
  t\notin\sem{s}\,\vee\,\exists u,m:\,u\in\sem{s}\land\sigma(u,m)>0}}
  \hspace{-2em}\mathsf{present}_\sigma(s,\tau,n)}^{\text{required affinity}} 
\]

In particular, the definition captures the exception for self-affinity that
avoids the bootstrap problem for scheduling self-affine pods, by allowing the
first such pod to be scheduled when no matching pod is yet present in the
cluster~\cite{k8s-assigning}. Formally, for an affinity term $(s,\tau)$ such
that the incoming type $t$ itself matches the selector, i.e., $t\in\sem{s}$, the
term is not imposed when no pod matching $s$ is present anywhere in the cluster.
In this way, the first pod of a self-affine group can be scheduled despite the
absence of an existing affinity witness. Once a pod matching $s$ exists in the
cluster, the exception no longer applies and the usual affinity condition is
enforced: a matching pod must be present in the topology domain $D_\tau(n)$ of
the candidate node. Thus, the cluster-wide test is used only to determine
whether the self-affinity bootstrap exception applies, whereas
$\mathsf{present}_\sigma(s,\tau,n)$ enforces the ordinary topology-local
affinity requirement.

\subsection{Topology-Spread Constraints}

Topology-spread constraints provide another mechanism for relating the placement
of an incoming pod to the distribution of existing pods, and specify how to
spread replicas across topology domains such as nodes, zones, or other
administrator-defined domains. In particular, a topology-spread constraint
restricts pod placement by rejecting candidate nodes that would cause exceeding
the permitted skew between topology domains.

We represent a hard topology-spread constraint by
\(
h=(s,\tau,k,\ell,E),
\)
where:
\(s\) is the selector identifying the pods whose distribution is counted;
\(\tau\) is the topology key;
\(k\geq1\) is the maximum permitted skew;
\(\ell\geq1\) is the minimum number of eligible topology domains; and
\(E\subseteq\Nodes\) is the set of eligible nodes.

For a topology value \(v\), define
\[
\mathsf{count}_\sigma(s,\tau,v,E)
=
\sum_{u\in\sem{s}}
\sum_{\substack{
m\in E\\
\lambda_N(m)(\tau)=v}}
\sigma(u,m).
\]
This is the number of resident pods matching \(s\) in the eligible part of
the \(\tau\)-\(v\) topology domain.

Given a topology-spread constraint
\(
h=(s,\tau,k,\ell,E),
\)
the skew calculation uses the minimum number of matching pods among eligible
domains defined as follows:
\[
\mathsf{gmin}_\sigma(h)
=
\begin{cases}
\displaystyle
\min_{v\in\{\lambda_N(n)(\tau)\mid n\in E\}}
\mathsf{count}_\sigma(s,\tau,v,E),
&
\text{if }
\left|
\{\lambda_N(n)(\tau)\mid n\in E\}
\right|
\geq\ell,
\\[2ex]
0,
&
\text{otherwise}.
\end{cases}
\]

Thus, when at least \(\ell\) eligible domains exist, the global minimum is
the minimum number of matching pods among them. When fewer than \(\ell\)
domains exist, Kubernetes treats the global minimum as zero.

Whether the incoming pod contributes to the count depends on whether its
type matches the selector. Define
\[
\delta(t,s)
=
\begin{cases}
1,&t\in\sem{s},\\
0,&t\notin\sem{s}.
\end{cases}
\]

Placing a pod of type \(t\) on node \(n\) satisfies the topology-spread
constraint \(h=(s,\tau,k,\ell,E)\) exactly when
\[
\mathsf{SpreadOK}(h,t,n,\sigma)
\iff
n\in E
\land
\mathsf{count}_\sigma
  (s,\tau,\lambda_N(n)(\tau),E)
+\delta(t,s)
-\mathsf{gmin}_\sigma(h)
\leq k.
\]

For example, suppose that two eligible zones \(a\) and \(b\) contain,
respectively, two and one pods matching \(s\), and that the maximum skew is
\(k=1\). If the incoming pod also matches \(s\), the global minimum is one.
Placing the pod in \(a\) would produce
\(
2+1-1=2>1
\)
and is therefore rejected. Placing it in \(b\), however, produces
\(
1+1-1=1,
\)
which satisfies the constraint.

For each pod type \(t\), let \(\mathsf{Spread}(t)\) denote its set of hard
topology-spread constraints.

\subsection{Hard Scheduling Feasibility}
\label{subsec:schedulability}

We can now combine the preceding Kubernetes constraints into a single
predicate describing the filtering stage relevant to our analysis.

\begin{definition}[Feasible deployment]
\label{def:hard}
A pod of type \(t\) is feasible on node \(n\) in configuration \(\sigma\) if \(\Hard(t,n,\sigma)\) holds, such that
\begin{align*}
\Hard(t,n,\sigma)\iff{}&
\used(\sigma,n)+\texttt{res}(t)\leq C(n)
\\
&{}\land\
\StaticOK(t,n)
\\
&{}\land\
\mathsf{AffinityOK}(t,n,\sigma)
\\
&{}\land\
\bigwedge_{h\in\mathsf{Spread}(t)}
\mathsf{SpreadOK}(h,t,n,\sigma).
\end{align*}
\end{definition}

The four components of this predicate correspond to the aspects of hard
Kubernetes scheduling represented by our model: resource capacity, static node
compatibility, required inter-pod affinity and anti-affinity, and hard topology
spread.

We conclude with an observation that will be useful for investigating the
complexity of the pod-deployability problem. Under the standard explicit
representation of nodes, pod types, labels, resources, selectors, and scheduling
constraints, evaluating \(\Hard(t,n,\sigma)\) takes polynomial time in the size
of the cluster representation. In fact, it is sufficient to check the
component-wise inequality \(\used(\sigma,n)+\texttt{res}(t)\leq C(n)\), check
\(\StaticOK(t,n)\) that we consider encoded in the cluster representation,
compute $\mathsf{present}_\sigma(s,\tau,n)$ to check
$\mathsf{AffinityOK}(t,n,\sigma)$ (in case of self-affinity, we also check
whether there is already a pod deployed), and checks \(n\in E_h\) for every
encoded constraint \(h\in\mathsf{Spread}(t)\). 
Note that comparison of binary-encoded integers (necessary to check resource
bounds and spread constraints) is polynomial in their bit length, so capacities
may be numerically exponential in the size of their encoding without affecting
the polynomial bound. 

\subsection{Scheduling as a Transition System}

The predicate in \cref{def:hard} describes whether a single placement is legal.
To reason about sequences of scheduling decisions, deployments, and deletions,
we model the cluster as a transition system whose states are configurations.

A successful scheduling operation adds one pod instance to one node. Since we
abstract away the scoring stage that introduces only soft constraints, the
choice among feasible nodes is nondeterministic: any node satisfying \(\Hard\)
may be selected.

We also permit deletion of a running pod, which abstracts how pods may disappear
from a Kubernetes cluster---at the level of our model, the precise operational
cause for a deletion, like scale-downs, evictions, and terminations, is irrelevant.

The transition relation is generated by the following two rules:
\[
\frac{\Hard(t,n,\sigma)}
     {\sigma\longrightarrow\sigma+(t,n)}
\;\textsc{Deploy}
\qquad
\frac{\sigma(t,n)>0}
     {\sigma\longrightarrow\sigma-(t,n)}
\;\textsc{Delete}.
\]

The \textsc{Deploy} rule abstracts filtering, selection, and successful binding.
The rule states only that the placement is permitted by all hard constraints
represented in the model.

An important consequence of this semantics is that scheduling feasibility is not
preserved under later deletions. Indeed, affinity and topology-spread
constraints are evaluated at a pod's \textsc{Deploy}ment and removing a pod can
invalidate the condition that enabled the placement.

Dually, scheduling unfeasibility does not persist under later deletions.
For example, consider a single node with a resource capacity of two and three
pod types \(a\), \(b\), and \(c\), each requesting one unit of the resource.
Suppose that \(b\) requires affinity to \(a\), while \(c\) requires affinity to
\(b\). Starting from the empty configuration, the sequence
\[
\zero
\longrightarrow
\{a\}
\longrightarrow
\{a,b\}
\longrightarrow
\{b\}
\longrightarrow
\{b,c\}
\]
is legal: first deploy \(a\), then use \(a\) as the affinity witness for \(b\);
delete \(a\); and finally use \(b\) as the witness for \(c\). At no point can
all three pods coexist because the node has capacity two. Nevertheless, the
temporary presence of \(a\) enables \(b\), which remains after \(a\) is deleted
and can subsequently enable \(c\). Thus, the set of pods present in the final
configuration does not contain all the intermediate witnesses needed to
construct it.

This distinction between \emph{currently satisfied constraints} and
\emph{constraints satisfied at deployment time} is central to the reachability
behavior of the model. In particular, this deploy-remove mechanism underlies the
hardness proofs in \cref{sec:logical-pspace,sec:capacity-pspace}.

\section{Pod-Deployability}
\label{sec:problem}

In this section, we formulate the verification problem studied throughout the
paper. We focus on the basic question of whether one pod type can ever be
deployed on one designated node. This property is sufficient to express, for
example, whether a sensitive workload may cross a placement boundary and run
on an untrusted node. More elaborate properties could require several pods to
be present simultaneously or could prescribe an entire final cluster
configuration. We leave such variants aside and study the simplest
node-specific deployment question, which already exhibits the complexity
boundaries developed in the following sections.

\begin{definition}[Pod-deployability]
Let \(\sigma_0\) be an initial configuration, \(t_\star\in\PodTypes\) a target
pod type, and \(n_\star\in\Nodes\) a target node. The target pod is
\emph{deployable} on \(n_\star\) from \(\sigma_0\) if \(\Deploy(\sigma_0,t_\star,n_\star)\) holds, where
\[
\Deploy(\sigma_0,t_\star,n_\star)
\iff
\exists\sigma:
\sigma_0\longrightarrow^*\sigma
\land
\sigma(t_\star,n_\star)\geq1
\]
where $\longrightarrow^*$ is the reflexive and transitive closure of
$\longrightarrow$.
The decision problem \textnormal{\textsc{Pod-Deployability}} asks whether this condition
holds.
\end{definition}

The definition requires only that the final configuration contains at least one
instance of \(t_\star\) on \(n_\star\). It imposes no condition on the other
pods or nodes. Pod-deployability is therefore a \emph{coverability} problem:
the target specifies a lower bound on one coordinate of the configuration.

Although an execution witnessing pod-deployability may be exponentially long,
the problem can be solved in polynomial space whenever every configuration has
a polynomial-size encoding.
We make the required boundedness assumption
explicit.

\begin{definition}[Capacity-bounded instance]
\label{def:capacityBounded}
An instance is \emph{capacity-bounded} if, for every pod type
\(t\in\PodTypes\) and node \(n\in\Nodes\), there exists a resource dimension
\(r\in\Resources\) such that
$
\texttt{res}(t)(r)>0.
$
The capacity constraint gives the finite bound
\[
B_{t,n}
=
\min_{\substack{r\in\Resources\\\texttt{res}(t)(r)>0}}
\left\lfloor
\frac{C(n)(r)}{\texttt{res}(t)(r)}
\right\rfloor
\]
on the number of simultaneously running instances of type \(t\) on node
\(n\).
\end{definition}

This is a natural assumption for Kubernetes deployments. Every running pod
consumes some finite node resource, even when its manifest omits explicit CPU
or memory requests: it requires runtime and bookkeeping resources and is
subject to node-level limits, such as the maximum number of pods supported by
a node. 
The assumption does not bound how many instances may be created over the
course of an execution. A controller may repeatedly create, delete, and
recreate pods. It bounds only the number of instances represented in one
configuration. 

The numerical value of \(B_{t,n}\) may be exponential in the
input length, but its binary representation is polynomial in that length.
This implies that 
a configuration can be stored in polynomial space. For each
pair \((t,n)\in\PodTypes\times\Nodes\), the capacity constraint implies
$
0\leq\sigma(t,n)\leq B_{t,n}.
$
We store \(\sigma(t,n)\) as a binary counter. This counter requires
$
\left\lceil\log_2(B_{t,n}+1)\right\rceil
$
bits. Since capacities and requests are binary encoded, this number is
polynomial in the input length. There are
\(|\PodTypes|\cdot|\Nodes|\) counters, and therefore the complete
configuration has a polynomial-size representation.

The number of possible configurations is finite and bounded by
\[
K=
\prod_{t\in\PodTypes}
\prod_{n\in\Nodes}
(B_{t,n}+1).
\]
Although \(K\) may be exponential, its binary
length is
polynomial in the input length.

If a target configuration
is reachable, then it is reachable by a path of length less than $K$
that does not repeat any configuration.
Hence, the deployability of a pod can be checked by a nondeterministic
algorithm which guesses at most $K$ transitions.
Notice that this algorithm is in $\mathsf{NPSPACE}$ because it has
to store the current configuration and the counter of the already guessed transitions.
Thanks to Savitch's theorem \cite{savitch1970relationships}, stating that 
\(\mathsf{NPSPACE}=\PSPACE\), we can conclude that 
\textsc{Pod-Deployability} is in \(\PSPACE\)
for capacity-bounded instances.

\begin{theorem}[Polynomial-space pod-deployability]
\label{thm:deployability-pspace}
\textnormal{\textsc{Pod-Deployability}} for capacity-bounded instances is in
\(\PSPACE\).
\end{theorem}

\section{A Polynomial Fragment without Affinity}
\label{sec:polynomial}

We now show that it is sufficient to remove affinity to obtain a fragment where
\textnormal{\textsc{Pod-Deployability}} can be solved in polynomial time.
Formally, we focus on \(\mathsf{AAS}\) be the fragment where pods do not declare
affinity constraints, i.e., : $A^+(t)=\varnothing$ for every $t\in\PodTypes$.
Note that anti-affinity constraints are admitted, i.e., it is possible to have
$t\in\PodTypes$ such that $A^-(t)\neq\varnothing$.

The key property of \(\mathsf{AAS}\) is that if a pod of type \(t\) is feasible
on node \(n\) in a configuration \(\sigma\), then it is feasible also in the
empty configuration \(\zero\).

According to \cref{def:hard}, \(t\) is feasible on node \(n\) in a configuration
\(\sigma\) if $\StaticOK(t,n)$, $\used(\sigma,n)+\texttt{res}(t)\leq C(n)$,
$\mathsf{AffinityOK}(t,n,\sigma)$, and $\bigwedge_{h\in\mathsf{Spread}(t)}
\mathsf{SpreadOK}(h,t,n,\sigma)$. $\StaticOK(t,n)$ is independent of the current
configuration. The inequation $\used(\sigma,n)+\texttt{res}(t)\leq C(n)$ implies
$\texttt{res}(t)\leq C(n)$, because $\used(\sigma,n)$ is a nonnegative vector.
$\mathsf{AffinityOK}(t,n,\zero)$ holds because affinity requirements are not
admitted and anti-affinity requirements are vacuously satisfied in the empty
configuration. Moreover, $\bigwedge_{h\in\mathsf{Spread}(t)}
\mathsf{SpreadOK}(h,t,n,\sigma)$ implies $\bigwedge_{h\in\mathsf{Spread}(t)}
\mathsf{SpreadOK}(h,t,n,\zero)$ because for every topology-spread constraint
\(h=(s,\tau,k,\ell,E)\), $\mathsf{SpreadOK}(h,t,n,\sigma)$ implies $n \in E$,
$\mathsf{gmin}_{\zero}(h)=0$, and in the configuration $\zero$ the spread
inequality reduces to
$
0+\delta(t,s)-0\leq k
$
which holds because $\delta(t,s) \leq 1$ and $k \geq 1$. This implies that a pod
of type \(t\) is feasible on node \(n\) also in the empty configuration
\(\zero\).

The above key property of \(\mathsf{AAS}\) allows us to check
$\Deploy(\sigma_0,t_\star,n_\star)$ simply by verifying whether node \(n_\star\)
already hosts a pod of type \(t_\star\) in the initial configuration
\(\sigma_0\) or $\Hard(t_\star,n_\star,\zero)$ holds.
In fact, if the pod is deployed after a sequence of transitions we have that
$\Hard(t_\star,n_\star,\sigma)$ holds for the configuration $\sigma$ in which
the pod is deployed; as observed above, this implies
$\Hard(t_\star,n_\star,\zero)$. Conversely, if $\Hard(t_\star,n_\star,\zero)$
holds, from the initial configuration $\sigma_0$ we can consider a sequence of
delete transitions that remove all the pods, thus reaching the empty
configuration $\zero$, from which a pod of type \(t_\star\) can be deployed on
node \(n_\star\). %
The above verification procedure, i.e., check the initial configuration or
$\Hard(t_\star,n_\star,\zero)$, takes polynomial time (see the discussion at the
end of \cref{subsec:schedulability}); hence
\textnormal{\textsc{Pod-Deployability}} in \(\mathsf{AAS}\) can be verified in
polynomial time.

\begin{theorem}[Polynomial affinity-free pod-deployability]
\label{thm:poly}
\textnormal{\textsc{Pod-Deployability}} in \(\mathsf{AAS}\) is decidable in polynomial time. 
\end{theorem}

\section{PSPACE with Affinity and Anti-Affinity}
\label{sec:logical-pspace}

We now isolate the expressive power of inter-pod affinity constraints. In this section,
we consider a capacity-free fragment in which the only active scheduling
constraints are required inter-pod affinity and %
anti-affinity. All resource requests are zero, all static node filters are
true, and no topology-spread constraints are present. The reduction uses a
single node, and therefore all affinity and anti-affinity terms use the
hostname topology. 

We show that this fragment is already \(\PSPACE\)-complete by encoding
coverability of \emph{1-safe} Petri nets. We first recall the source problem, then
present the encoding and its intuition and argue that it is polynomial.

\noindent
\textbf{1-safe Petri nets.}
A 1-safe net is a Petri net in which every reachable marking contains at most one
token in each place. Formally, a 1-safe net
is a tuple
$
\mathcal P=(P,T,\mathsf{Pre},\mathsf{Post},M_0),
$
where \(P\) and \(T\) are finite sets of places and transitions. For each
transition \(a\in T\), the sets
\(\mathsf{Pre}(a),\ \mathsf{Post}(a)\subseteq P\) are its input and output
places. The initial marking \(M_0\subseteq P\) indicates places
that initially contain one token.
In 1-safe nets, every reachable marking contains at most one
token in each place and can therefore be identified with the set of places
\(M\subseteq P\) containing a token.
A transition \(a\) is enabled at \(M\) when
$
\mathsf{Pre}(a)\subseteq M.
$
Firing \(a\) consumes the tokens in its input places and produces tokens in
its output places:
$
M\xrightarrow{a}
(M\setminus\mathsf{Pre}(a))
\cup\mathsf{Post}(a).
$
The source problem used in the reduction is \emph{goal coverability}: given a
1-safe Petri net and a distinguished place \(g\in P\), determine whether some
reachable marking contains \(g\). 
This problem is \(\PSPACE\)-complete~\cite{cheng1995safe}.\footnote{Goal coverability
is a simplified version, but with the same complexity, of the classical coverability 
problem for 1-safe Petri nets: coverability considers a set of target places instead of only one goal place.}

\noindent
\textbf{Encoding.}
Let
$
\mathcal{P}=(P,T,\mathsf{Pre},\mathsf{Post},M_0)
$
be a 1-safe Petri net with distinguished goal place \(g\). We construct a
Kubernetes instance consisting of a single node. All pod types have zero
resource requests.
A token in a place \(p\) is represented by the presence of a pod carrying the
label \(\mathsf{place}=p\). Several pod types may carry the same label
\(\mathsf{place}=p\). For every transition
\(a\in T\), we partition the places involved in its firing into
$\mathsf{Consumed}(a)$ (i.e., $\mathsf{Pre}(a)\setminus\mathsf{Post}(a)$),
$\mathsf{Preserved}(a)$ (i.e., $\mathsf{Pre}(a)\cap\mathsf{Post}(a)$),
and $\mathsf{Produced}(a)$ (i.e., $\mathsf{Post}(a)\setminus\mathsf{Pre}(a)$).
Tokens in \(\mathsf{Consumed}(a)\) are removed during the simulation of \(a\),
tokens in \(\mathsf{Preserved}(a)\) remain present throughout, and tokens in
\(\mathsf{Produced}(a)\) are newly created.

The simulation is driven by a four-phase
\emph{start-execute-end-ready} protocol for each transition. A
\(\mathsf{ready}\) label marks a stable control state in which a new
transition may be started. The initial configuration contains one initial
ready pod and one initial token pod \(I_p\) for every \(p\in M_0\). Both
initial types require affinity to a fresh label \(\bot\) that no pod type
carries, so they can be deleted but never redeployed.

For every transition \(a\), a start type \(S_a\) requires affinity to a ready
pod and to a token pod for every \(p\in\mathsf{Pre}(a)\), which is deployable
when \(a\) is enabled. Its anti-affinity to the global \(\mathsf{start}\) and
\(\mathsf{end}\) labels serializes protocols. Once \(S_a\) is deployed, the
ready pod and the tokens in \(\mathsf{Consumed}(a)\) are deleted; tokens in
\(\mathsf{Preserved}(a)\) are left untouched. An execute type \(E_a\) then
requires affinity to \(S_a\) and anti-affinity to ready and to every consumed
place, so it appears only after the consumed inputs have been removed. While
\(E_a\) is present, an output-token type \(P_{p,a}\) for every
\(p\in\mathsf{Produced}(a)\) becomes deployable, so that it carries
\(\mathsf{place}=p\) and is anti-affine to \(\mathsf{place}=p\), ensuring at
most one token per place. An end type \(F_a\) requires affinity to \(E_a\) and
to every place in \(\mathsf{Post}(a)\), certifying that all outputs are present.
Finally, a ready type \(R_a\) requires affinity to \(F_a\) and anti-affinity to
the start and execute labels, so it appears only after the management pods of
the finished protocol have been deleted, and restores a stable configuration
whose token projection is exactly
\((M\setminus\mathsf{Pre}(a))\cup\mathsf{Post}(a)\).

A final subtlety is that Kubernetes permits arbitrary pod deletions, so a
computation may lose tokens.
This is not a problem: Petri-net enabling is monotone in the
marking, so any execution interleaved with token losses can be lifted to a
loss-free execution reaching a larger marking.
Thus, if the goal place is marked by an execution with token losses,
it will be marked also by the corresponding loss-free execution.

The construction introduces a polynomial number of pod types and constraints:
one initial token type per initially marked place, one output-token type for
every transition-output-place pair, and a constant number of control types
(start, execute, end, ready) per transition, plus one goal type. Labels and
selectors use indices from \(P\cup T\). The encoding function is therefore
computable in time polynomial in the size of \(\mathcal P\).

Finally, notice that although the fragment has no resource capacity, pod-deployability remains in
\(\PSPACE\). Affinity and anti-affinity inspect only whether a matching pod is
present; they do not inspect its exact multiplicity. Consequently, a
configuration can be abstracted by its \emph{support}
$
\mathsf{supp}(\sigma)
=
\{(t,n)\in\PodTypes\times\Nodes\mid \sigma(t,n)>0\}.
$
Adding another instance of an already present type,
or deleting one instance when another remains, does not
change any affinity or anti-affinity guard.
Multiplicities can therefore be
saturated at one, and there are at most
\(2^{|\PodTypes|\cdot|\Nodes|}\) support configurations. A nondeterministic
algorithm stores the current support, guesses one deployment or deletion
transition, checks its guards in polynomial time, and searches for at most
this many configurations. The algorithm uses polynomial space, and hence
pod-deployability belongs to \(\mathsf{NPSPACE}=\PSPACE\) by Savitch's
theorem \cite{savitch1970relationships}.

\begin{theorem}[Affinity and anti-affinity fragment]
\label{thm:logical-pspace}
\textnormal{\textsc{Pod-Deployability}} is \(\PSPACE\)-complete when the only active scheduling
constraints are required inter-pod affinity and required inter-pod
anti-affinity. Hardness already holds for one node.
\end{theorem}

\section{PSPACE with Affinity and One Capacity}
\label{sec:capacity-pspace}

The preceding section showed that required affinity and anti-affinity alone
already yield \(\PSPACE\)-completeness, without relying on resource
capacities. We now consider a
complementary fragment: anti-affinity is prohibited, and the only dynamic
constraints are required inter-pod affinity and one scalar resource capacity.
We can show that this fragment is also \(\PSPACE\)-complete, even if we can use a single node and unit resource requests. Thus, neither anti-affinity nor
multiple resources, nodes, or topology domains are needed for the lower bound.

We show that this fragment is already \(\PSPACE\)-complete by encoding the
bounded black pebbling game.
We first recall the source problem, then
present the encoding and its intuition and argue that it is polynomial.

\noindent
\textbf{Bounded black pebbling game.}
Let \(G=(V,E)\) be a finite directed acyclic graph (DAG). For a vertex \(v\in V\), let
$
\mathsf{pred}(v)=\{u\in V\mid (u,v)\in E\}
$
be the set of its immediate predecessors. A \emph{pebble configuration} is a set
\(X\subseteq V\) such that a vertex belongs to \(X\) when it currently carries a
pebble.
The black-pebble game starts from the empty configuration. It permits the
following moves.
\begin{itemize}
\item A pebble may be placed on an unpebbled vertex \(v\) when every immediate
predecessor of \(v\) is currently pebbled:
$
v\notin X $ and $
\mathsf{pred}(v)\subseteq X.
$
The resulting configuration is \(X\cup\{v\}\).

\item A pebble may be removed from any %
vertex \(v\in X\). The resulting
configuration is \(X\setminus\{v\}\).
\end{itemize}
Source vertices have no predecessors and may therefore be pebbled at any time.
Removing a predecessor does not remove pebbles that were placed using it; it
may only prevent future placements until that predecessor is reconstructed.

In the bounded game, every configuration must contain at most \(k\) pebbles.
An instance of \emph{bounded black-pebble coverability} consists of a DAG
\(G=(V,E)\), a target vertex \(z\in V\), and a bound \(k\). The question is
whether the empty configuration can reach some pebble configuration containing
\(z\), without ever using more than \(k\) pebbles. This problem is
\(\PSPACE\)-complete \cite{gilbert1979pebbling}.

\noindent
\textbf{Encoding.}
Let \((G,z,k)\) be a bounded black-pebble coverability instance, with
\(G=(V,E)\). We construct a Kubernetes instance with a single node \(n\) and
one scalar resource of capacity \(C(n)=k\). For every vertex \(v\in V\) we
introduce one pod type \(P_v\) that requests one unit of the resource and
carries the unique label \(\mathsf{id}=v\). For every edge \((u,v)\in E\), the
type \(P_v\) has one required %
pod-affinity term whose selector matches
only \(\mathsf{id}=u\); hence, \(P_v\) can be deployed only while an
instance of every predecessor type \(P_u\) with \(u\in\mathsf{pred}(v)\) is
running. Source types have no affinity terms and are deployable whenever one
unit of capacity is free. The initial configuration is empty, and the
pod-deployability target is \((P_z,n)\).

The intuition is direct. The presence of an instance of \(P_v\) represents a
pebble on \(v\): deploying the first \(P_v\) instance is pebble placement, deleting 
its last instance is pebble removal, and the node capacity enforces the pebble bound.
Required affinity expresses %
the predecessor prerequisite, and, because
affinity is checked only at scheduling time, a pod may outlive the
prerequisites that enabled it, mirroring that removing a predecessor does not
remove pebbles placed using it.

A final subtlety is that, in the absence of anti-affinity constraints,
Kubernetes permits contemporaneous deployments of several instances of
the pods of type \(P_v\).
This is not a problem: any sequence of pod deployments and deletions
which deploys the target pod $P_z$, has a corresponding correct
sequence obtained by removing the deployments of instances of 
already present \(P_v\) pods.

The construction uses one node, one scalar capacity, one pod type per vertex,
and one affinity term per edge, so its size is \(O(|V|+|E|+\log k)\) and it is
computable in polynomial time.
Note that the produced instance is capacity-bounded hence
pod-deployability is in \(\PSPACE\) by \cref{thm:deployability-pspace}.

\begin{theorem}[Affinity and One-capacity fragment]
\label{thm:capacity-pspace}
\textnormal{\textsc{Pod-Deployability}} is \(\PSPACE\)-complete 
for a capacity-bounded fragment when %
the only scheduling
constraint is inter-pod required affinity. Hardness already holds for one node,
one scalar resource and unit pod requests.
\end{theorem}

\section{Related Work and Conclusion}
\label{sec:related}
\label{sec:conclusion}

Formal models of Kubernetes have previously been proposed for purposes other than reasoning about scheduling reachability. In particular, Turin et al.~\cite{DBLP:journals/jss/TurinBDDJT23} develop a model of Kubernetes in Real-Time ABS to capture the resource-sensitive behaviour of containerized microservices and predict CPU and memory consumption under different deployment scenarios. Their abstraction is aimed at resource modelling and performance prediction, rather than at formalizing the scheduler's placement decisions and the constraints that determine whether a pod can be assigned to a node. More generally, existing work on Kubernetes formalization does not, to the best of our knowledge, study the computational complexity induced by the scheduler's placement semantics. Our work takes a complementary perspective: we formalize the hard scheduling constraints that govern pod placement and their evolution under deployments and deletions, and study the resulting verification problem. To the best of our knowledge, this is the first formal analysis of Kubernetes pod-deployability and of its computational complexity.

At first sight, anti-affinity appears to be the natural source of hardness:
it is a negative test, since deployment is enabled by the \emph{absence} of
matching workloads.  This intuition is consistent with classical
computational models, in which adding zero or absence tests can substantially
increase expressive power.  For example, ordinary Petri-net coverability is
\(\mathsf{EXPSPACE}\)-hard~\cite{lipton1976reachability},
but it becomes undecidable 
in the presence of inhibitor arcs able to test the absence of tokens~\cite{hack1976petri}.  Our results
show that negative tests do not imply a similar complexity jump
in pod scheduling.  On the other hand, required affinity, which is 
capable of checking the presence of selected pods, suffices 
for \(\PSPACE\)-hardness already when combined with one scalar capacity.  
The work closest to ours is the analysis of APP and aAPP by De Palma et
al.~\cite{depalma2026ifm}.  APP is a language for specifying scheduling
policies for Function-as-a-Service platforms, and aAPP extends it with notions
of affinity and anti-affinity similar to ours.  Their reachability
question---whether a function can be scheduled on a designated worker---is
therefore closely related to pod-deployability.  They prove a linear-time
bound for the anti-affinity-only fragment and
\(\PSPACE\)-completeness for the full language containing both affinity and
anti-affinity, using a reduction from propositional
planning~\cite{bylander1994planning}.  For the affinity-only fragment,
however, they establish NP-hardness but leave its exact complexity open.
Our results sharpen and extend this picture in two directions.  First, our
formal language reflects Kubernetes placement and supports node-label
restrictions, multidimensional resources, topology-aware affinity and
anti-affinity, and hard topology-spread constraints.  In particular, affinity
is not restricted to the candidate node: a pod may require a matching witness
anywhere in the node's zone or in another labelled topology domain.  Second,
our lower bounds already hold for deliberately minimal fragments.  The
reduction from 1-safe Petri nets establishes \(\PSPACE\)-hardness using only
required affinity and anti-affinity, without capacities or spread
constraints.  More significantly, the bounded-pebbling reduction establishes
\(\PSPACE\)-hardness with required affinity alone, on a single node and with a
single scalar capacity.  Thus, for our scheduler model, it closes the
affinity-only complexity gap identified for aAPP and shows that anti-affinity
is not necessary for \(\PSPACE\)-hardness.

Pod-deployability naturally lends itself to a formulation as a coverability
problem, since it asks whether some reachable configuration contains at least
one instance of a target pod type on a designated node, without constraining the
remainder of the cluster. Exact reachability would instead prescribe the
complete target configuration, including the absence and multiplicity of every
other pod type. This distinction parallels that between coverability and
reachability in Petri nets, where the two problems may have substantially
different complexities~\cite{cheng1995safe,lipton1976reachability}. An
intermediate and practically relevant generalization is \emph{co-occurrence}:
can two designated pod types be simultaneously present on the same node? This is
a coverability query over two type-node coordinates and corresponds closely to
the co-occurrence property studied for aAPP scheduling policies by De Palma et
al.~\cite{depalma2026ifm}. Our lower bounds extend immediately to this problem
by adding an unconstrained marker pod and asking whether it co-occurs with the
original target. Moreover, the polynomial-space upper bound is preserved:
capacity-bounded configurations still have polynomial-size representations,
while the capacity-free affinity/anti-affinity fragment can still be explored
through its finite support abstraction. We therefore expect co-occurrence to be
\(\PSPACE\)-complete in both of our hard fragments. Exact configuration
reachability deserves a separate analysis. When capacities bound every pod
multiplicity, the configuration graph is finite and each state has a
polynomial-size representation, so the same exploration argument gives a
\(\PSPACE\) upper bound. Without binding capacities, however, arbitrarily many
instances of a zero-request pod type may be deployed on the same node. Our
support abstraction identifies all positive multiplicities: configurations
containing one, two, or one million instances of a type have the same support
and satisfy the same affinity and anti-affinity constraints. This identification
is sound for pod-deployability because that problem only asks whether at least
one target instance is present. It is not sound for exact reachability, which
must distinguish, for example, a target configuration containing exactly one
instance from one containing two. Consequently, the finite support graph used in
our \(\PSPACE\) argument loses information needed by exact reachability. Whether
these unbounded multiplicities can be handled by a different finite abstraction,
and what complexity results, are left for future work.

Our study also relates to the broader cluster-scheduling literature, which
provides the architectural setting but generally asks different questions.
Kubernetes descends from large-scale cluster managers such as
Borg~\cite{verma2015borg}. In this area, schedulers like Omega investigated
shared-state optimistic scheduling~\cite{schwarzkopf2013omega}, while Firmament
formulated placement as a min-cost-flow problem~\cite{gog2016firmament}---Burns
et al.~\cite{burns2016borg} provide a detailed description of the lineage from
Borg and Omega to Kubernetes. This line of work primarily concerns scheduler
architecture, placement quality, usage, and scalability. Instead, we study the
complexity of an existential verification question induced by Kubernetes' hard
filters. The Kubernetes documentation specifies the scheduling framework and the
placement constraints abstracted by our semantics
\cite{k8s-framework,k8s-assigning}. Our model isolates the hard scheduling-time
conditions relevant to pod-deployability and deliberately omits scoring,
transient failures, preemption, and plugin-internal state.

\noindent
\textbf{Conclusion.}
We introduced pod-deployability as the coverability problem of determining
whether some legal sequence of deployments and deletions can place a pod of a
given type on a designated node.  Our results identify a sharp complexity
boundary.  Without affinity and without binding capacity, pod-deployability is
decidable in polynomial time, even in the presence of anti-affinity and hard
topology-spread constraints.  With required affinity and anti-affinity it is
\(\PSPACE\)-complete even without capacities.  Finally, with required affinity
alone it remains \(\PSPACE\)-complete on one node with one scalar capacity.
Together, these results show that positive placement dependencies, rather than
only negative exclusion tests, are sufficient to make long deployment and
deletion histories computationally significant.

For security-sensitive placement policies, inspecting each %
constraint in
isolation is therefore insufficient: an unsafe placement may become feasible
only after a long sequence of creations and deletions, even when it is absent
from the current cluster state.  At the same time, the polynomial fragment
shows that efficient verification is possible when no positive prerequisite
must be maintained.   A natural direction for future work is to incorporate
cluster autoscaling, allowing nodes to be dynamically added and removed.
Because this yields an unbounded state space, it raises new questions about the
decidability and complexity of pod-deployability.

\bibliographystyle{plainurl}
\bibliography{references}
\newpage
\appendix

\section{Proofs of Pod-Deployability without affinity constraints}
\label{app:deployability}

In this appendix, we give the full proof of \cref{thm:deployability-pspace},
which states that \textsc{Pod-Deployability} is in \(\PSPACE\) for
capacity-bounded instances. We first show that every configuration of a
capacity-bounded instance admits a polynomial-size representation
(\cref{lem:poly-configuration}), then bound the number of distinct
configurations and the length of shortest witnessing executions
(\cref{lem:configuration-count,lem:short-witness}), and finally combine these
facts in a nondeterministic polynomial-space decision procedure
(\cref{thm:deployability-pspace-proof}).

\begin{lemma}[Polynomial-size configurations]
\label{lem:poly-configuration}
For a capacity-bounded instance, every configuration \(\sigma\) can be stored
in polynomial space.
\end{lemma}

\begin{proof}
For each pair \((t,n)\in\PodTypes\times\Nodes\), the capacity constraint
implies
\[
0\leq\sigma(t,n)\leq B_{t,n}.
\]
We store \(\sigma(t,n)\) as a binary counter. This counter requires
\[
\left\lceil\log_2(B_{t,n}+1)\right\rceil
\]
bits. Since capacities and requests are encoded in binary, this number is
polynomial in the input length. There are
\(|\PodTypes|\cdot|\Nodes|\) counters, and therefore the complete
configuration has a polynomial-size representation.
\end{proof}

\begin{lemma}[Finite configuration count]
\label{lem:configuration-count}
The number of possible configurations is finite and bounded by
\[
K=
\prod_{t\in\PodTypes}
\prod_{n\in\Nodes}
(B_{t,n}+1).
\]
Although \(K\) may be exponential, its binary length is polynomial in the
input length.
\end{lemma}

\begin{proof}
Each counter \(\sigma(t,n)\) ranges over the \(B_{t,n}+1\) values
\(0,\dots,B_{t,n}\), so the product counts all configurations. Its binary
length is
\[
\left\lceil\log_2 K\right\rceil
\leq
\sum_{t\in\PodTypes}
\sum_{n\in\Nodes}
\left\lceil\log_2(B_{t,n}+1)\right\rceil,
\]
which is polynomial in the input length by \cref{lem:poly-configuration}.
\end{proof}

\begin{lemma}[Short witnessing paths]
\label{lem:short-witness}
If a target configuration is reachable, then it is reachable by a path of
length less than \(K\) that does not repeat any configuration.
\end{lemma}

\begin{proof}
Take a shortest path from \(\sigma_0\) to a target configuration. If some
configuration occurred twice, the cycle between its two occurrences could be
removed, yielding a strictly shorter path; this contradicts minimality.
Hence the path is simple. Since there are at most \(K\) distinct
configurations by \cref{lem:configuration-count}, a simple path has length
strictly less than \(K\).
\end{proof}

\begin{theorem}[Polynomial-space pod-deployability]
\label{thm:deployability-pspace-proof}
\textnormal{\textsc{Pod-Deployability}} for capacity-bounded instances is in
\(\PSPACE\).
\end{theorem}

\begin{proof}
Consider the following nondeterministic procedure. It stores the current
configuration, initially \(\sigma_0\), together with a binary step counter
bounded by \(K\). At every iteration, it first checks whether
\[
\sigma(t_\star,n_\star)\geq1.
\]
If so, it accepts. Otherwise, it nondeterministically chooses one of the
following transitions:
\begin{itemize}
\item choose \(t\in\PodTypes\) and \(n\in\Nodes\), verify
\(\Hard(t,n,\sigma)\), and replace \(\sigma\) by \(\sigma+(t,n)\); or
\item choose \(t\in\PodTypes\) and \(n\in\Nodes\) with
\(\sigma(t,n)>0\), and replace \(\sigma\) by \(\sigma-(t,n)\).
\end{itemize}
If no accepting configuration has been found after \(K\) transitions, the
branch rejects.

Each transition changes one binary counter by one. For the constraints of
\cref{sec:model}, resource sums, selector matching, affinity and
anti-affinity tests, and topology-spread counts can all be computed by
scanning the polynomially many counters, nodes, labels, and encoded
constraints. Therefore the predicate \(\Hard(t,n,\sigma)\) can be checked
in polynomial space. The procedure stores only the current configuration,
the step counter, and polynomially many auxiliary values, so it uses
polynomial space.

By \cref{lem:short-witness}, a target configuration is reachable exactly
when it is reachable within \(K\) transitions, so the procedure accepts
exactly the positive instances. Thus \textsc{Pod-Deployability} belongs to
\(\mathsf{NPSPACE}=\PSPACE\) by Savitch's theorem
\cite{savitch1970relationships}.
\end{proof}

\section{The Encoding of 1-Safe Petri Nets}
\label{app:encodings}
\label{app:petri}
\label{app:petri-encoding}

In this appendix, we formalize the polynomial-time encoding function
\(\mathcal E_{\mathrm{PN}}\) used in \cref{sec:logical-pspace} and prove its
correctness: the goal place of a 1-safe Petri net is coverable if and only if
the goal pod is deployable in the encoded instance.

We use the following notation throughout this and the next appendix. For a
label predicate \(L\), write \(s[L]\) for the selector that matches exactly
the pod types carrying label \(L\). All affinity and anti-affinity terms use
the hostname topology of the unique node; we therefore write only their
selectors. Labels described as fresh are distinct from all other labels
introduced by the construction.

Let
$
  \mathcal P=(P,T,\mathsf{Pre},\mathsf{Post},M_0)
$
be a 1-safe Petri net and let \(g\in P\) be its goal place. We identify a
reachable marking with the set \(M\subseteq P\) of its marked places. A
transition \(a\in T\) is enabled at \(M\) when
\(\mathsf{Pre}(a)\subseteq M\), and its firing produces
$
  a(M)=
  (M\setminus\mathsf{Pre}(a))\cup\mathsf{Post}(a).
$
As recalled in \cref{sec:logical-pspace}, coverability of \(g\) is
\(\PSPACE\)-hard for 1-safe nets.

\paragraph*{The encoding function.}
We define
$
  \mathcal E_{\mathrm{PN}}(\mathcal P,g)
  =
  (\PodTypes_{\mathcal P},\{n\},\sigma_0,(G_g,n)).
$
The result is a pod-deployability instance with one node \(n\), zero resource
requests, no capacity restriction, and only required affinity and required
anti-affinity. The distinguished goal type \(G_g\) requires a token for
place \(g\) and is the pod-deployability target. For every transition
\(a\in T\), partition its places into consumed, preserved, and produced:
\[
  \mathsf{Consumed}(a)=\mathsf{Pre}(a)\setminus\mathsf{Post}(a),\ \
  \mathsf{Preserved}(a)=\mathsf{Pre}(a)\cap\mathsf{Post}(a),\ \
  \mathsf{Produced}(a)=\mathsf{Post}(a)\setminus\mathsf{Pre}(a).
\]
The set of pod types is
\[
\PodTypes_{\mathcal P}={}
 \{I_p\mid p\in M_0\} 
 \ \cup\ \{P_{p,a}\mid a\in T,\ p\in\mathsf{Produced}(a)\}
 \ \cup\ \{S_a,E_a,F_a,R_a\mid a\in T\}
 \ \cup\ \{R_0,G_g\}.
\]
Here \(I_p\) is the initial token type for \(p\), \(P_{p,a}\) is a token for
\(p\) produced by transition \(a\), and \(S_a,E_a,F_a,R_a\) are the
start, execute, end, and ready types of the protocol for \(a\). The type
\(R_0\) supplies the initial ready pod.

Let \(\bot\) be a fresh label that no pod type carries. The labels and hard
constraints generated by \(\mathcal E_{\mathrm{PN}}\) are as follows; every
unlisted affinity or anti-affinity set is empty.
\begin{description}
\item[Initial token \(I_p\).]
It carries \(\mathsf{place}=p\), requires affinity to \(s[\bot]\), and has
anti-affinity to \(s[\mathsf{place}=p]\). Thus its initial instance can be
deleted but the type can never be deployed again.

\item[Initial ready type \(R_0\).]
It carries \(\mathsf{ready}\), requires affinity to \(s[\bot]\), and has
anti-affinity to \(s[\mathsf{ready}]\). It likewise cannot be recreated.

\item[Start type \(S_a\).]
It carries \(\mathsf{start}\) and \(\mathsf{start}=a\). Its affinity terms
are
\[
 \{s[\mathsf{ready}]\}
 \cup
 \{s[\mathsf{place}=p]\mid p\in\mathsf{Pre}(a)\},
\]
and its anti-affinity terms are
\(s[\mathsf{start}]\) and \(s[\mathsf{end}]\).

\item[Execute type \(E_a\).]
It carries \(\mathsf{execute}\) and \(\mathsf{execute}=a\), requires affinity
to \(s[\mathsf{start}=a]\), and has anti-affinity to
\[
 \{s[\mathsf{ready}],s[\mathsf{execute}]\}
 \cup
 \{s[\mathsf{place}=p]\mid p\in\mathsf{Consumed}(a)\}.
\]
The absence test is imposed only on the consumed places; preserved places
(self-loops) are left untouched.

\item[Produced-token type \(P_{p,a}\).]
It carries \(\mathsf{place}=p\), requires affinity to
\(s[\mathsf{execute}=a]\), and has anti-affinity to
\(s[\mathsf{place}=p]\).

\item[End type \(F_a\).]
It carries \(\mathsf{end}\) and \(\mathsf{end}=a\). It requires
affinity to
\[
 \{s[\mathsf{execute}=a]\}
 \cup
 \{s[\mathsf{place}=p]\mid p\in\mathsf{Post}(a)\},
\]
and has anti-affinity to \(s[\mathsf{end}]\).

\item[Ready type \(R_a\).]
It carries \(\mathsf{ready}\), requires affinity to
\(s[\mathsf{end}=a]\), and has anti-affinity to
\(s[\mathsf{ready}]\), \(s[\mathsf{start}]\), and
\(s[\mathsf{execute}]\).

\item[Goal type \(G_g\).]
It carries the fresh label \(\mathsf{goal}\) and has the single affinity term
\(s[\mathsf{place}=g]\). Thus it is deployable exactly when a token for
\(g\) is present.
\end{description}
All types have zero resource requests and pass every static filter. The
initial configuration is
\[
 \sigma_0(I_p,n)=1\quad(p\in M_0),
 \qquad \sigma_0(R_0,n)=1,
\]
and is zero elsewhere. Required resident anti-affinity ensures that at most
one token carrying \(\mathsf{place}=p\), one ready pod, and one pod carrying
each global control label can be present.

For a Kubernetes configuration \(\sigma\), define its token projection by
\[
 \pi_{\mathrm{PN}}(\sigma)
 =\{p\in P\mid
     \text{some running pod in \(\sigma\) carries \(\mathsf{place}=p\)}\}.
\]
A configuration is \emph{stable for \(M\)} if its token projection is \(M\),
it contains exactly one ready pod, and it contains no start, execute, or end
pod. Let \(\Phi_{\mathrm{PN}}(M)\) denote any such stable configuration. The
choice of producer type for a token is irrelevant because all token types for
the same place carry the same place label and the same uniqueness constraint.

\begin{theorem}[Polynomiality of the Petri-net encoding]
\label{thm:petri-encoding-poly}
The encoding function \(\mathcal E_{\mathrm{PN}}\) is computable in time
polynomial in the size of \(\mathcal P\).
\end{theorem}

\begin{proof}
The encoding creates one node, one initial token type per initially marked
place, one produced-token type per transition-produced-place pair, four
control types per transition, one initial ready type, and one goal type. It
contains one affinity term per preset or postset test and only a constant
number of additional guards per type. Labels and selectors use indices from
\(P\cup T\), so the total size and the computation time are polynomial in
\(|P|+|T|+|M_0|\).
\end{proof}

\begin{theorem}[PSPACE membership of the capacity-free fragment]
\label{thm:logical-pspace-membership}
Pod-deployability is in \(\PSPACE\) when the only active scheduling
constraints are required inter-pod affinity and required inter-pod
anti-affinity, even without any capacity bound.
\end{theorem}

\begin{proof}
Affinity and anti-affinity inspect only whether a matching pod is present, not
its exact multiplicity. Abstract a configuration by its support
\(\mathsf{supp}(\sigma)=\{(t,n)\mid\sigma(t,n)>0\}\). Adding another instance
of an already present type satisfies no new affinity term, and deleting one
instance while another remains changes no guard; hence multiplicities saturate
at one. There are at most \(2^{|\PodTypes|\cdot|\Nodes|}\) supports. A
nondeterministic algorithm \emph{a}) stores the current support and a step
counter bounded by this number, \emph{b}) guesses one deployment or deletion,
and \emph{c}) checks its guards in polynomial time. The search uses polynomial
space, so the problem is in \(\mathsf{NPSPACE}=\PSPACE\) by Savitch's theorem
\cite{savitch1970relationships}.
\end{proof}

\begin{lemma}[One firing is simulated]
\label{lem:petri-completeness-step}
If \(M\xrightarrow{a}M'\) in \(\mathcal P\), then every stable configuration
for \(M\) reaches a stable configuration for \(M'\).
\end{lemma}

\begin{proof}
Because \(a\) is enabled,
\(\mathsf{Pre}(a)\subseteq M\). Hence the ready pod and the token pods for
all places in \(\mathsf{Pre}(a)\) witness the affinity terms of \(S_a\).
No start or end pod is present in a stable configuration, so its
anti-affinity terms also hold and \(S_a\) can be deployed.

Delete the ready pod and precisely the token pods for
\(\mathsf{Consumed}(a)\); the tokens for \(\mathsf{Preserved}(a)\) remain.
The start pod remains as a lock: its resident anti-affinity prevents every
other start and every end. Now \(E_a\) is deployable. Its positive witness
is \(S_a\), and its negative tests certify that ready and every consumed
input token have disappeared. Once \(E_a\) is present, deploy one
\(P_{p,a}\) for every \(p\in\mathsf{Produced}(a)\). Its execute affinity
holds. Its place anti-affinity also holds: if
\(p\in\mathsf{Produced}(a)=\mathsf{Post}(a)\setminus\mathsf{Pre}(a)\),
1-safety implies that \(p\notin M\) whenever \(a\) is enabled at a reachable
marking.

After all outputs exist, delete \(S_a\) and deploy \(F_a\). The end
affinities certify that \(E_a\) and every postset token are present: a
preserved place is witnessed by the token kept throughout and a produced
place by the newly deployed \(P_{p,a}\). Delete \(E_a\), then deploy
\(R_a\), witnessed by \(F_a\). Finally delete \(F_a\). The remaining
configuration contains one ready pod, no control pod, and exactly the tokens
\((M\setminus\mathsf{Pre}(a))\cup\mathsf{Post}(a)=M'\). It is therefore
stable for \(M'\). The sequence is illustrated in
\cref{fig:petri-encoding}.
\end{proof}

\begin{figure}[t]
\centering
\begin{tikzpicture}[
  >=Latex,
  place/.style={circle,draw,minimum size=5.5mm,inner sep=0pt},
  trans/.style={rectangle,draw,minimum width=3mm,minimum height=8mm},
  state/.style={rounded corners,draw,align=center,minimum width=39mm,
                minimum height=8mm,font=\small},
  lab/.style={font=\scriptsize,align=center}
]

\node[place] (p) at (0,-0.925) {$p$};
\node[place] (q) at (0,-1.925) {$q$};
\node[trans] (a) at (1.25,-1.425) {};
\node[place] (r) at (2.5,-1.425) {$r$};
\draw[->] (p) -- (a);
\draw[->] (q) -- (a);
\draw[->] (a) -- (r);
\node[lab] at (1.25,-2.225) {transition $a$};
\node[lab] at (1.25,-0.425) {Petri net};

\node[state] (s0) at (7.1,1.75)
  {$R,\ P_p,\ P_q$\\stable marking $\{p,q\}$};
\node[state,below=3.5mm of s0] (s1)
  {$S_a,\ P_p,\ P_q$\\start: test the inputs};
\node[state,below=3.5mm of s1] (s2)
  {$S_a,\ E_a$\\execute: consumed inputs gone};
\node[state,below=3.5mm of s2] (s3)
  {$E_a,\ P_r$\\produce every output};
\node[state,below=3.5mm of s3] (s4)
  {$F_a,\ P_r$\\end, then recreate $R$};
\node[state,below=3.5mm of s4] (s5)
  {$R,\ P_r$\\stable marking $\{r\}$};
\draw[->] (s0) -- (s1);
\draw[->] (s1) -- node[right,lab] {delete $R,P_p,P_q$} (s2);
\draw[->] (s2) -- (s3);
\draw[->] (s3) -- node[right,lab] {delete $S_a$} (s4);
\draw[->] (s4) -- node[right,lab] {delete $E_a,F_a$} (s5);
\node[lab] at (7.1,2.35) {Kubernetes protocol};

\draw[decorate,decoration={brace,amplitude=8pt,mirror}]
  (4.5,2.5) -- (4.5,-5.40);

\draw[->,dashed] (3,-1.45) -- (4,-1.45);
\end{tikzpicture}
\caption{Encoding a Petri-net transition with
\(\mathsf{Pre}(a)=\{p,q\}\) and \(\mathsf{Post}(a)=\{r\}\).
The displayed configurations suppress producer indices and transient deletion
steps. The ready pod \(R\) separates completed protocol cycles.}
\label{fig:petri-encoding}
\end{figure}
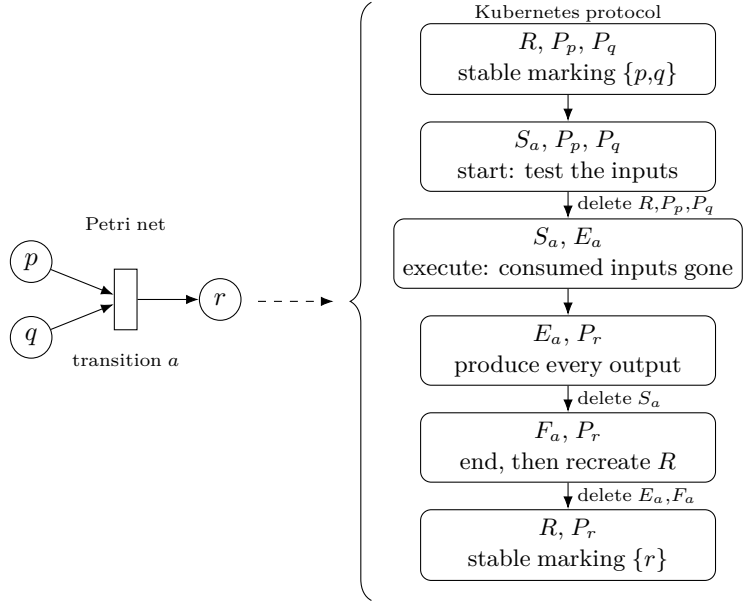

\begin{corollary}[Completeness of \(\mathcal E_{\mathrm{PN}}\)]
\label{cor:petri-completeness}
If \(M_0\xrightarrow{*}M\), then
\(\sigma_0\longrightarrow^*\Phi_{\mathrm{PN}}(M)\). In particular, if
\(g\in M\), the encoded pod-deployability instance can deploy \(G_g\).
\end{corollary}

\begin{proof}
The initial configuration is stable for \(M_0\). Apply
\cref{lem:petri-completeness-step} inductively to the firing sequence.
\end{proof}

Kubernetes also permits deletions not prescribed by the simulation. We use
the following standard monotonicity argument to show that such losses cannot
create a false positive.

\begin{lemma}[Lifting token losses]
\label{lem:petri-loss-lifting}
Suppose a sequence starting at \(M_0\) alternates Petri-net firings with
arbitrary losses \(L\mapsto L'\subseteq L\), and ends in \(L_f\). Then
\(\mathcal P\) has a loss-free firing sequence from \(M_0\) to a reachable
marking \(M_f\supseteq L_f\).
\end{lemma}

\begin{proof}
Maintain by induction a reachable loss-free marking \(M\) containing the
current lossy marking \(L\). Initially take \(M=L=M_0\). A loss preserves
the invariant without changing \(M\). If \(a\) is enabled at \(L\), then
\(\mathsf{Pre}(a)\subseteq L\subseteq M\), so \(a\) is also enabled at
\(M\). Since \(M\) is reachable and the net is 1-safe, firing \(a\) is a
legal safe firing. Moreover,
\[
 (L\setminus\mathsf{Pre}(a))\cup\mathsf{Post}(a)
 \subseteq
 (M\setminus\mathsf{Pre}(a))\cup\mathsf{Post}(a).
\]
Taking the right-hand side as the new \(M\) preserves the invariant.
\end{proof}

\begin{lemma}[Protocol decomposition]
\label{lem:petri-protocol-decomposition}
Every Kubernetes execution generated by \(\mathcal E_{\mathrm{PN}}\) is a
concatenation of completed protocol intervals, followed by at most one
unfinished interval. Each completed interval projects to one enabled
Petri-net firing, possibly together with arbitrary token losses. An unfinished
interval cannot start another protocol.
\end{lemma}

\begin{proof}
A start pod requires ready. Once some \(S_a\) is deployed, its resident
anti-affinity to \(s[\mathsf{start}]\) excludes every other start pod.
Moreover, ready must be deleted before \(E_a\) can appear. Execute requires
the particular \(S_a\) and certifies, by anti-affinity, that all consumed
preset tokens have been deleted. Thus reaching execute identifies an enabled
transition \(a\) at the token projection observed when \(S_a\) was deployed.

Produced-token types for \(a\) require \(E_a\). An end pod requires all
postset tokens and cannot appear while another end pod remains. Finally, a
new ready pod requires \(F_a\) and cannot appear while a start or execute pod
remains. A subsequent start is blocked until the end pod is also deleted.
Consequently the only way to return to a stable control state is the ordered
cycle
\[
  \mathsf{ready}\ ;\ S_a\ ;\ E_a\ ;\
  \{P_{p,a}:p\in\mathsf{Produced}(a)\}\ ;\ F_a\ ;\ R_a.
\]
Deleting a required control pod early makes completion impossible, because
deleted control types cannot be recreated without an earlier phase. Such an
execution is therefore an unfinished final interval and cannot enable a new
start. Deletions of unrelated token pods are the only remaining deviations;
they project exactly to token losses.
\end{proof}

\begin{theorem}[Correctness of the Petri-net encoding]
\label{thm:petri-encoding-correct}
The goal place \(g\) is coverable from \(M_0\) in \(\mathcal P\) if and only
if the target \((G_g,n)\) is coverable in
\(\mathcal E_{\mathrm{PN}}(\mathcal P,g)\).
\end{theorem}

\begin{proof}
The forward implication is \cref{cor:petri-completeness}. For the reverse
implication, consider the execution prefix ending immediately before the first
deployment of \(G_g\). Its affinity guard shows that this prefix contains a
token for \(g\). By \cref{lem:petri-protocol-decomposition}, every completed
protocol in the prefix projects to a Petri firing, interspersed with possible
token losses. If the goal token is already present after a completed
interval, \cref{lem:petri-loss-lifting} yields a genuine loss-free execution
whose final marking contains \(g\).

Otherwise the goal token is produced in the unique unfinished final interval.
It must be an output \(P_{g,a}\). Before it can appear, \(S_a\) was deployed
while all preset tokens were present, the consumed tokens and ready were
deleted, and \(E_a\) was deployed. Hence \(a\) was enabled in the lossy
marking at the start of the interval. Append the abstract firing of \(a\);
its postset contains \(g\). Applying \cref{lem:petri-loss-lifting} again
yields a genuine Petri execution covering \(g\). Thus no partial or lossy
Kubernetes execution creates a false positive.
\end{proof}

Together with \cref{thm:petri-encoding-poly,thm:logical-pspace-membership},
\cref{thm:petri-encoding-correct} proves \cref{thm:logical-pspace}.

\paragraph*{Lean formalization.}
The encoding and its correctness are mechanized in the Lean~4 library in
\texttt{lean\_code} and available at
\url{https://doi.org/10.5281/zenodo.22025656}. The scheduler model
(configurations, affinity and anti-affinity semantics, deployment and deletion
steps, reachability, and pod-deployability) is defined in
\texttt{Scheduler.lean}. One-safe Petri nets, their firing semantics,
coverability, and the token-loss lifting lemma (\cref{lem:petri-loss-lifting})
are formalized in \texttt{PetriNet.lean} and \texttt{PetriNetEncoding.lean}; the
lossy step relation \texttt{LossyNetStep} combines an ordinary firing with an
arbitrary token loss, and \texttt{loss\_lifting} shows that every lossy abstract
execution is simulated by a loss-free execution whose marking contains the lossy
marking.

The pod roles are an inductive type with constructors for initial tokens,
produced tokens, start, execute, end, ready, and the goal; selectors such as
\texttt{tokenSelector}, \texttt{readySelector}, and \texttt{executeEqSelector}
implement the label predicates, and the functions \texttt{affinity} and
\texttt{antiAffinity} give the concrete guards described above. The initial-only
pods are blocked from redeployment by the affinity to the empty selector
\texttt{botSelector} (\(s[\bot]\)), exactly as in the definition of \(\mathcal
E_{\mathrm{PN}}\).

Completeness (\cref{cor:petri-completeness}) is proved in
\texttt{PetriNetCompleteness.lean}. The proof constructs the scheduler witness
for one simulated firing as an explicit sequence of legal add and remove
operations---deploy start, delete ready and the consumed inputs, deploy execute,
deploy the produced outputs, delete start, deploy end, delete execute, deploy
ready, delete end---and composes these sequences along the Petri firing sequence
by induction, in \texttt{stableConfig\_transition\_to\_stableConfig} and
\texttt{petri\_covers\_implies\_goal\_deployable}.

Soundness (reverse direction of \cref{thm:petri-encoding-correct}) is proved in
\texttt{PetriNetCorrectness.lean}. The proof first establishes that every
reachable configuration satisfies a valid control combination
(\texttt{reachable\_validControlCombination}): each control role is globally
unique, ready cannot coexist with execute, and start cannot coexist with end. It
then splits every scheduler execution at its stable configurations using the
inductively defined \texttt{FirstStableTrace} and \texttt{SinceClosestStable}
decompositions, shows that every completed stable interval corresponds to one
enabled Petri firing or to pure token loss
(\texttt{first\_stable\_reach\_sound}), and that the final interval producing
the goal token requires an enabled transition whose postset contains the goal
place (\texttt{preset\_present\_at\_closest\_stable\_before\_output}). Combining
these with the loss-lifting lemma yields
\texttt{goal\_deployable\_implies\_petri\_covers}. The main equivalence,
\texttt{petri\_covers\_iff\_goal\_deployable} in
\texttt{PetriNetMainTheorem.lean}, combines the two directions and
corresponds exactly to \cref{thm:petri-encoding-correct}.

\section{The Encoding of Bounded Black Pebbling}
\label{app:pebbling}
\label{app:pebble-encoding}

In this appendix, we formalize the polynomial-time encoding function
\(\mathcal E_{\mathrm{Peb}}\) used in \cref{sec:capacity-pspace} and prove
its correctness: the target vertex of a bounded black-pebble instance is
coverable with at most \(k\) pebbles if and only if the target pod is
deployable in the encoded instance.

Let \((G,z,k)\) be a bounded black-pebble coverability instance, where
\(G=(V,E)\) is a DAG, \(z\in V\), and \(k\) is the pebble bound. Write
\(\mathsf{pred}(v)=\{u\mid(u,v)\in E\}\). A pebble position is a set
\(X\subseteq V\) with \(|X|\leq k\). A move either removes any \(v\in X\),
or adds \(v\notin X\) when \(\mathsf{pred}(v)\subseteq X\) and
\(|X|<k\).

\paragraph*{The encoding function.}
Define
\[
 \mathcal E_{\mathrm{Peb}}(G,z,k)
 =
 (\{P_v\mid v\in V\},\{n\},\zero,(P_z,n)).
\]
The unique node has one scalar resource \(r\) with capacity \(C(n)(r)=k\).
Every type \(P_v\) has unit request
\(\texttt{res}(P_v)(r)=1\), carries the unique label
\(\mathsf{id}=v\), and has required affinity terms
\[
 A^+(P_v)=
 \{s[\mathsf{id}=u]\mid u\in\mathsf{pred}(v)\}.
\]
All anti-affinity sets are empty and all static filters are true. Thus a
source type has no affinity terms and can be deployed whenever capacity is
available. The construction uses one type per vertex and one affinity term
per edge.

Define the canonical encoding of a pebble position \(X\) by
\[
 \Phi_{\mathrm{Peb}}(X)(P_v,n)=
 \begin{cases}
 1,&v\in X,\\
 0,&v\notin X,
 \end{cases}
\]
and project an arbitrary Kubernetes configuration to its support:
\[
 \pi_{\mathrm{Peb}}(\sigma)
 =\{v\in V\mid\sigma(P_v,n)>0\}.
\]
Notice that
\(|\pi_{\mathrm{Peb}}(\sigma)|\leq
  \sum_v\sigma(P_v,n)\leq k\).

\begin{theorem}[Polynomiality of the pebbling encoding]
\label{thm:pebble-encoding-poly}
The encoding function \(\mathcal E_{\mathrm{Peb}}\) is computable in time
polynomial in the size of \((G,z,k)\).
\end{theorem}

\begin{proof}
The encoding creates one node, one scalar capacity, one pod type per vertex,
and one affinity term per edge. Its size is \(O(|V|+|E|+\log k)\), and it can
be computed in the same time.
\end{proof}

\begin{theorem}[PSPACE membership of the affinity-and-capacity fragment]
\label{thm:capacity-pspace-membership}
Pod-deployability is in \(\PSPACE\) for a fragment with one node, one scalar
resource capacity, unit pod requests, and required inter-pod affinity only.
\end{theorem}

\begin{proof}
A fragment where each pod type requires one unit of a unique
available resource contains problem instances that are capacity-bounded 
(see \cref{def:capacityBounded}). Pod-deployability is then
in \(\PSPACE\) as a consequence of \cref{thm:deployability-pspace-proof}.
\end{proof}

\begin{lemma}[Completeness of \(\mathcal E_{\mathrm{Peb}}\)]
\label{lem:pebble-completeness}
If the bounded pebble game moves from \(X\) to \(Y\), then
\(\Phi_{\mathrm{Peb}}(X)\longrightarrow
  \Phi_{\mathrm{Peb}}(Y)\). Consequently, every bounded pebbling execution
from \(\varnothing\) has a corresponding Kubernetes execution from \(\zero\).
\end{lemma}

\begin{proof}
Suppose first that \(Y=X\cup\{v\}\). Legality of the pebble placement gives
\(\mathsf{pred}(v)\subseteq X\) and \(|X|<k\). For every predecessor
\(u\), configuration \(\Phi_{\mathrm{Peb}}(X)\) contains \(P_u\), so every
affinity term of \(P_v\) is satisfied. Its resource usage is
\(|X|+1\leq k\), so capacity also holds. Deploying \(P_v\) yields exactly
\(\Phi_{\mathrm{Peb}}(Y)\), as illustrated in
\cref{fig:pebble-encoding}. If \(Y=X\setminus\{v\}\), deleting the unique
instance of \(P_v\) yields \(\Phi_{\mathrm{Peb}}(Y)\). Induction on the move
sequence proves the second claim.
\end{proof}

\begin{figure}[t]
\centering
\begin{tikzpicture}[
  >=Latex,
  vertex/.style={circle,draw,minimum size=7mm,inner sep=0pt},
  pebbled/.style={vertex,fill=black,text=white},
  pod/.style={rounded corners,draw,fill=blue!8,minimum width=10mm,
              minimum height=6mm,font=\small},
  lab/.style={font=\small,align=center}
]
\node[pebbled] (u) at (0,1.2) {$u$};
\node[pebbled] (w) at (0,0) {$w$};
\node[vertex] (v) at (1.7,0.6) {$v$};
\draw[->] (u) -- (v);
\draw[->] (w) -- (v);
\node[lab] at (0.8,1.85) {position $X=\{u,w\}$};

\node[pebbled] (u1) at (0,-2.2) {$u$};
\node[pebbled] (w1) at (0,-1) {$w$};
\node[pebbled] (v1) at (1.7,-1.6) {$v$};
\draw[->] (u1) -- (v1);
\draw[->] (w1) -- (v1);
\node[lab] at (0.8,-2.85) {position $X \cup \{v\}=\{u,w,v\}$};

\draw[->] (1,-.1) -- (1,-.7);

\draw[->,very thick] (2.35,0.6) -- node[above,lab]
  {$\Phi_{\mathrm{Peb}}$} (3.45,0.6);
\node[draw,rounded corners,minimum width=34mm,minimum height=18mm]
  (before) at (5.25,0.6) {};
\node[lab] at (5.25,1.28) {node $n$, capacity $k=3$};
\node[pod] at (4.65,0.48) {$P_u$};
\node[pod] at (5.85,0.48) {$P_w$};
\node[lab] at (5.25,-0.05) {position $X$};

\node[draw,rounded corners,minimum width=34mm,minimum height=16mm]
  (after) at (5.25,-1.45) {};
\draw[->,very thick] (2.35,-1.45) -- node[above,lab]
  {$\Phi_{\mathrm{Peb}}$} (3.35,-1.45);
\node[pod] at (4.15,-1.45) {$P_u$};
\node[pod] at (5.25,-1.45) {$P_w$};
\node[pod,fill=green!12] at (6.35,-1.45) {$P_v$};
\draw[->] (before) -- node[right,lab]
  {deploy $P_v$; affinity to $P_u,P_w$} (after);
\node[lab] at (5.25,-2.0) {position $X\cup\{v\}$;};
\end{tikzpicture}
\caption{A pebble position and its canonical pod configuration. Since both
predecessor pods are present and one capacity unit is free, deploying \(P_v\)
implements the legal move \(X\mapsto X\cup\{v\}\).}
\label{fig:pebble-encoding}
\end{figure}
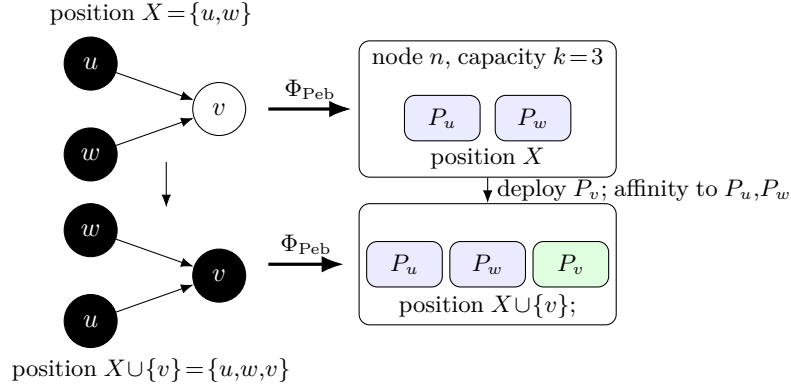

Because the affinity-only language cannot forbid duplicate instances, the
reverse direction must apply to noncanonical configurations as well.

\begin{lemma}[Support simulation]
\label{lem:pebble-support-simulation}
For every Kubernetes transition \(\sigma\longrightarrow\sigma'\), exactly one
of the following holds:
\begin{enumerate}
\item \(\pi_{\mathrm{Peb}}(\sigma')=
      \pi_{\mathrm{Peb}}(\sigma)\); or
\item \(\pi_{\mathrm{Peb}}(\sigma)\longrightarrow
      \pi_{\mathrm{Peb}}(\sigma')\) is a legal bounded pebble move.
\end{enumerate}
\end{lemma}

\begin{proof}
Consider a deployment of \(P_v\). If \(P_v\) was already present, support
does not change. Otherwise the support gains \(v\). Every affinity term was
checked before deployment, hence
\(\mathsf{pred}(v)\subseteq\pi_{\mathrm{Peb}}(\sigma)\). Moreover, capacity
requires
\[
  \sum_{u\in V}\sigma(P_u,n)<k.
\]
Since support counts at most one vertex per
running instance,
\(|\pi_{\mathrm{Peb}}(\sigma)|<k\). Adding \(v\) is therefore a legal
bounded pebble move.

Now consider deletion of one \(P_v\). If another instance remains, support
is unchanged. If the deleted pod was the last instance, support loses \(v\),
which is always a legal pebble-removal move. These cases exhaust the
transition relation.
\end{proof}

\begin{corollary}[Soundness of \(\mathcal E_{\mathrm{Peb}}\)]
\label{cor:pebble-soundness}
Every Kubernetes execution from \(\zero\) projects, after deleting
support-preserving stuttering steps, to a legal bounded pebbling execution
from \(\varnothing\).
\end{corollary}

\begin{proof}
Apply \cref{lem:pebble-support-simulation} to each transition. The projected
initial support is empty, every non-stuttering projected step is a legal
pebble move, and every projected support has size at most \(k\).
\end{proof}

\begin{theorem}[Correctness of the pebbling encoding]
\label{thm:pebble-encoding-correct}
The target vertex \(z\) is coverable with at most \(k\) black pebbles if and
only if \((P_z,n)\) is pod-deployable in
\(\mathcal E_{\mathrm{Peb}}(G,z,k)\).
\end{theorem}

\begin{proof}
If a bounded pebbling execution reaches \(X\ni z\),
\cref{lem:pebble-completeness} produces a Kubernetes execution reaching
\(\Phi_{\mathrm{Peb}}(X)\), which contains \(P_z\). Conversely, suppose a
Kubernetes execution reaches \(\sigma\) with \(\sigma(P_z,n)>0\). Then
\(z\in\pi_{\mathrm{Peb}}(\sigma)\). By
\cref{cor:pebble-soundness}, this support is reachable by a bounded pebbling
execution, after omitting stuttering steps. Hence that pebbling execution
covers \(z\).
\end{proof}

Together with \cref{thm:pebble-encoding-poly,thm:capacity-pspace-membership},
\cref{thm:pebble-encoding-correct} proves \cref{thm:capacity-pspace}.

\paragraph*{Lean formalization.}
The pebbling encoding and its correctness are mechanized in the Lean~4 library
in \texttt{lean\_code}, availabile at
\url{https://doi.org/10.5281/zenodo.22025656}. The bounded black-pebble game is
defined in \texttt{Pebbling.lean}: a \texttt{PebbleGame} records the predecessor
function, the target, and the bound; \texttt{Step} is the inductive legal-move
relation (placement requires the predecessors to be present and spare capacity,
removal is always legal), and \texttt{Covers} asks for a reachable position
containing the target.

The encoding is formalized in \texttt{PebbleEncoding.lean}. The scheduler
specification \texttt{spec} gives each vertex pod a unit request, the node a
capacity equal to the bound, and one affinity term \texttt{predecessorTerm}
per edge; the initial configuration is empty. The canonical configuration of
a position is \texttt{canonical}, and the support projection is
\texttt{support}. Completeness (\cref{lem:pebble-completeness}) is proved by
\texttt{pebbleStep\_scheduler} and \texttt{pebbleRun\_scheduler}, which show
that every legal pebble move is simulated by a scheduler transition between
canonical configurations. Because affinity alone cannot forbid duplicate
instances, soundness is proved for arbitrary configurations:
\texttt{schedulerStep\_projects} shows that every scheduler step either
preserves the support (a stuttering step) or projects to a legal pebble move,
and \texttt{schedulerRun\_projects} lifts this to whole executions, also
showing that the resource usage never exceeds the bound. The main
equivalence, \texttt{encoding\_correct\_and\_complete}, combines the two
directions and corresponds to \cref{thm:pebble-encoding-correct}.

\end{document}